\documentclass[10pt]{llncs}
\usepackage{enumerate,amsmath,amssymb,mathrsfs}
\usepackage[all,dvips,color,arrow,curve]{xy}
\xyoption{arc}
\xyoption{poly}

\newcommand{\CSP}[0]{\ensuremath{\textsc{CSP}}}
\newcommand{\QCSP}[0]{\ensuremath{\textsc{QCSP}}}

\newcommand{\bA}[0]{\mathcal{A}}
\newcommand{\bB}[0]{\mathcal{B}}
\newcommand{\bC}[0]{\mathcal{C}}
\newcommand{\bK}[0]{\mathcal{K}}

\newcommand{\bH}[0]{\mathcal{H}}
\newcommand{\bP}[0]{\mathcal{P}}

\renewcommand{\qed}{$\Box$}
\renewenvironment{proof}
{\noindent{\bf Proof.}\ }
{\hfill\qed\par\bigskip}

\newenvironment{proofof}[1]
{\noindent{\bf Proof of #1.}\ }
{\hfill\qed\par\bigskip}

\begin{document}

\title{Constraint Satisfaction with Counting Quantifiers\thanks{The authors
gratefully acknowledge the facilities of the Fields Institute in Toronto where
the some of the work on this project was done during the Thematic Program of the
institute in Mathematics of Constraint Satisfaction in August
2011.}\vspace{-1ex}}

\author{
Barnaby Martin\inst{1}\thanks{Author supported by EPSRC grant EP/G020604/1.},
Florent Madelaine\inst{2} and
Juraj Stacho\inst{3}\thanks{Author supported by EPSRC grant EP/I01795X/1.}
}

\institute{
School of Engineering and Computing Sciences, Durham University\\ 
Science Laboratories, South Road, Durham DH1 3LE, UK.\\
\and
Clermont Universit\'{e}, Universit\'{e} d'Auvergne,\\
LIMOS, BP 10448, F-63000 Clermont-Ferrand, France.\\
\and
DIMAP and Mathematics Institute,\\
University of Warwick, Coventry CV4 7AL, UK.\\
\vspace{-3ex}
}

\maketitle

\begin{abstract}
We initiate the study of \,{\em constraint satisfaction problems}\, (CSPs) in the
presence of counting quantifiers, which may be seen as variants of CSPs in the
mould of \emph{quantified CSPs} (QCSPs).

\quad We show that a \textbf{single} counting
quantifier strictly between $\exists^{\geq 1}:=\exists$ and $\exists^{\geq
n}:=\forall$ (the domain being of size $n$) already affords the maximal possible
complexity of QCSPs (which have \textbf{both} $\exists$ and $\forall$), being
Pspace-complete for a suitably chosen template. 

\quad Next, we focus on the complexity of subsets of counting quantifiers on
clique and cycle templates. For cycles we give a full trichotomy -- all such
problems are in L, NP-complete or Pspace-complete.  For cliques we come close
to a similar trichotomy, but one case remains outstanding.

\quad Afterwards, we consider the generalisation of CSPs in which we augment the
extant quantifier $\exists^{\geq 1}:=\exists$ with the quantifier $\exists^{\geq
j}$ ($j \neq 1$). Such a CSP is already NP-hard on non-bipartite graph
templates. We explore the situation of this generalised CSP on bipartite
templates, giving various conditions for both tractability and hardness --
culminating in a classification theorem for general graphs. 

\quad Finally, we use
counting quantifiers to solve the complexity of a concrete QCSP whose complexity
was previously open.
\end{abstract}

\section{Introduction}
The \emph{constraint satisfaction problem} $\CSP(\mathcal{B})$, much studied in
artificial intelligence, is known to admit several equivalent formulations, two
of the best known of which are the query evaluation of primitive positive (pp)
sentences -- those involving only existential quantification and conjunction --
on $\mathcal{B}$, and the homomorphism problem to $\mathcal{B}$ (see, e.g.,
\cite{KolaitisVardiBook05}). The problem $\CSP(\mathcal{B})$ is NP-complete in
general, and a great deal of effort has been expended in classifying its
complexity for certain restricted cases. Notably it is conjectured
\cite{FederVardi,JBK} that for all fixed ${\cal B}$, the problem
$\CSP(\mathcal{B})$ is in P or NP-complete.  While this has not been settled in
general, a number of partial results are known -- e.g.  over structures of size
at most three \cite{Schaefer,Bulatov} and over smooth digraphs
\cite{HellNesetril,barto:1782}.  

A popular generalisation of the CSP involves considering the query evaluation
problem for \emph{positive Horn} logic -- involving only the two quantifiers,
$\exists$ and $\forall$, together with conjunction. The resulting
\emph{quantified constraint satisfaction problems} $\QCSP(\mathcal{B})$  allow
for a broader class, used in artificial intelligence to capture non-monotonic
reasoning, whose complexities rise to Pspace-completeness.

In this paper, we study counting quantifiers of the form $\exists^{\geq j}$,
which allow one to assert the existence of at least $j$ elements such that the
ensuing property holds. Thus on a structure $\mathcal{B}$ with domain of size
$n$, the quantifiers $\exists^{\geq 1}$ and $\exists^{\geq n}$ are precisely
$\exists$ and $\forall$, respectively. Counting quantifiers have been
extensively studied in finite model theory (see \cite{EbbinghausFlum,OttoBook}),
where the focus is on supplementing the descriptive power of various logics. Of
more general interest is the majority quantifier $\exists^{\geq n/2}$ (on a
structure of domain size $n$), which sits broadly midway between $\exists$ and
$\forall$. Majority quantifiers are studied across diverse fields of logic and
have various practical applications, \mbox{e.g.} in cognitive appraisal and
voting theory \cite{Voting}. They have also been studied in computational
complexity, \mbox{e.g.},~in~\cite{Lan04}.

We study variants of $\CSP(\mathcal{B})$ in which the input sentence to be
evaluated on $\mathcal{B}$ (of size $|B|$) remains positive conjunctive in its
quantifier-free part, but is quantified by various counting quantifiers from
some non-empty set. 

For $X \subseteq \{1,\ldots,|B|\}$, $X\neq \emptyset$, the $X$-CSP$(\bB)$ takes as input a sentence
given by a conjunction of atoms quantified by quantifiers of the form
$\exists^{\geq j}$ for $j \in X$. It then asks whether this sentence is true on
$\bB$. The idea to study $\{1,\ldots,|B|\}$-CSP$(\mathcal{B})$ is originally due
to Andrei Krokhin.

In Section~\ref{sec:single}, we consider the power of a single quantifier
$\exists^{\geq j}$. We prove that for each $n \geq 3$, there is a template
$\bB_n$ of size $n$, such that $\exists^{\geq j}$ ($1<j<n$) already has the full
complexity of QCSP, \mbox{i.e.}, $\{j\}$-CSP$(\bB_n)$ is Pspace-complete. 

In Section~\ref{sec:cliques_and_cycles}, we go on to study the complexity of
subsets of our quantifiers on clique and cycle templates, $\mathcal{K}_n$ and
$\mathcal{C}_n$, respectively. We derive the following classification theorems.
\vspace{-1ex}

\begin{theorem} \label{thm:cliques}
For $n\in \mathbb{N}$ and $X \subseteq \{1,\ldots,n\}$:\vspace{-1.5ex}
\begin{enumerate}[(i)]
\item
$X$-CSP$(\bK_n)$ is in L if $n \leq 2$ or $X\cap \big\{1,\ldots,\lfloor n/2
\rfloor\big\}=\emptyset$.
\item
$X$-CSP$(\bK_n)$ is NP-complete if $n>2$ and $X=\{1\}$.
\item
$X$-CSP$(\bK_n)$ is Pspace-complete if $n>2$ and either $j \in X$ for $1<j <
n/2$ or $\{1,j\} \subseteq X$ for $j\in\big\{\lceil n/2\rceil,\ldots
,n\big\}$.
\end{enumerate}
\end{theorem}

\noindent This is a near trichotomy -- only the cases where $n$ is even and we
have the quantifier $\exists^{\geq n/2}$ remain open. For cycles, however, the
trichotomy is complete. 

\begin{theorem}\label{thm:cycles}
For $n\geq 3$ and $X \subseteq \{1,\ldots,n\}$, the problem $X$-CSP$(\bC_n)$ is
either in L, is NP-complete or is Pspace-complete. Namely:\vspace{-1.5ex}
\begin{enumerate}[(i)]
\item $X$-CSP$(\bC_n)\in{\rm L}$ if $n=4$, or $1 \notin X$, or $n$ is even and
$X\cap \big\{2,\ldots,n/2\big\}=\emptyset$.
\item $X$-CSP$(\bC_n)$ is NP-complete if $n$ is odd and $X=\{1\}$.
\item $X$-CSP$(\bC_n)$ is Pspace-complete in all other cases.
\end{enumerate}
\end{theorem}

In Section~\ref{sec:extensions}, we consider $\{1,j\}$-CSP$(\bH)$, for $j \neq
1$ on graphs. The CSP is already NP-hard on non-bipartite graph templates. We
explore the situation of this generalised CSP on bipartite graph templates,
giving various conditions for both tractability and hardness, using and
extending results of Section~\ref{sec:cliques_and_cycles}. We are most
interested here in the distinction between P and NP-hard. To understand which of
these cases are Pspace-complete would include as a subclassification the
Pspace-complete cases of QCSP$(\mathcal{H})$, a question which has remained open
for five years \cite{CiE2006}.  We give a classification theorem for graphs in
fragments of the logic involving bounded use of $\exists^{\geq 2}$ followed by
unbounded use of $\exists$. In the case of QCSP ($\exists^{\geq n}$ instead of
$\exists^{\geq 2}$), this is perfectly natural and is explored with bounded
alternations in, e.g., \cite{HubieExRes,HubieNew,EdithH}, and with bounded use of
$\forall=\exists^{\geq n}$~in~\cite{hubie-sicomp}. We~prove that either there
exists such a fragment in which the problem is NP-hard or for all such
fragments the problem is in P.

Afterwards in Section~\ref{sec:new}, we use counting quantifiers to solve the
complexity of QCSP$(\bC^*_4)$, where $\bC^*_4$ is the
reflexive $4$-cycle, whose complexity was previously open.
Finally, in Section~\ref{sec:conclusion}, we give some closing remarks and
open problems.\vspace{-0.5ex}

\section{Preliminaries}\label{sec:prelims}

Let $\bB$ be a finite structure over a finite signature $\sigma$ whose domain
$B$ is of cardinality $|B|$. For $1 \leq j \leq |B|$, the formula $\exists^{\geq
j} x \ \phi(x)$ with \emph{counting quantifier} should be interpreted on $\bB$
as stating that there exist at least $j$ distinct elements $b \in B$ such that
$\bB \models \phi(b)$. Counting quantifiers generalise existential
($\exists:=\exists^{\geq 1}$), universal ($\forall:=\exists^{\geq |B|}$) and
(weak) majority ($\exists^{\geq |B|/2}$) quantifiers. Counting
quantifiers do not in general commute with themselves, viz $\exists^{\geq j} x
\exists^{\geq j} y \neq \exists^{\geq j} y \exists^{\geq j} x$ (in contrast,
$\exists$ and $\forall$ do commute with themselves, but not with one another). 

For $\emptyset \neq X \subseteq \{1,\ldots,|B|\}$, the $X$-CSP$(\bB)$ takes as input a sentence
of the form $\Phi:=Q_1 x_1 Q_2 x_2 \ldots Q_m x_m \ \phi(x_1,x_2,\ldots,x_m)$,
where $\phi$ is a conjunction of positive atoms of $\sigma$ and each $Q_i$ is of
the form $\exists^{\geq j}$ for some $j \in X$. The set of such sentences
forms the logic $X$-pp (recall the pp is primitive positive). The
yes-instances are those for which $\bB \models \Phi$. Note that all problems
$X$-CSP$(\bB)$ are trivially in Pspace, by cycling through all possible
evaluations for the variables.

The problem $\{1\}$-CSP$(\bB)$ is better-known as just CSP$(\bB)$, and
$\{1,|B|\}$-CSP$(\bB)$ is better-known as QCSP$(\bB)$. We will consider also the
logic $[2^m1^*]$-pp and restricted problem $[2^m1^*]$-CSP$(\bB)$, in which the
input $\{1,2\}$-pp sentence has prefix consisting of no more than $m$
$\exists^{\geq 2}$ quantifiers followed by any number of $\exists$ quantifiers
(and nothing else).

A homomorphism from $\bA$ to $\bB$, both $\sigma$-structures, is a function $h:A
\rightarrow B$ such that $(a_1,\ldots,a_r) \in R^\bA$ implies
$(h(a_1),\ldots,h(a_r)) \in R^\bB$, for all relations $R$ of~$\sigma$. A
frequent role will be played by the \emph{retraction} problem
Ret$(\mathcal{\bB})$ in which one is given a structure $\bA$ containing $\bB$,
and one is asked if there is a homomorphism from $\bA$ to $\bA$ that is the
identity on $\bB$. It is well-known that retraction problems are special
instances of CSPs in which the constants of the template are all named
\cite{FederHell98}.

In line with convention we consider the notion of hardness reduction in proofs
to be polynomial many-to-one (though logspace is sufficient for our results).

\subsection{Game characterisation}
\vspace{-1ex}

There is a simple game characterisation for the truth of sentences of the logic
$X$-pp on a structure $\bB$. Given a sentence $\Psi$ of $X$-pp, and a structure
$\bB$, we define the following game $\mathscr{G}(\Psi,\bB)$. Let $\Psi:=Q_1 x_1
Q_2 x_2 \ldots Q_m x_m \ \psi(x_1,x_2,\ldots,x_m)$. Working from the outside in,
coming to a quantified variable $\exists^{\geq j} x$, the \emph{Prover} (female)
picks a subset $B_x$ of $j$ elements of $B$ as witnesses for $x$, and an
\emph{Adversary} (male) chooses one of these, say $b_x$, to be the value of $x$.
Prover wins iff $\bB \models \psi(b_{x_1},b_{x_2},\ldots,b_{x_m})$. The
following comes immediately from the definitions.
\vspace{-0.7ex}

\begin{lemma}\label{lem:game}
Prover has a winning strategy in the game $\mathscr{G}(\Psi,\bB)$ iff $\bB
\models \Psi$.
\end{lemma}
\vspace{-0.7ex}

We will often move seemlessly between the two characterisations of
Lemma~\ref{lem:game}.  One may alternatively view the game in the language of
homomorphisms. There is an obvious bijection between $\sigma$-structures with
domain $\{1,\ldots,m\}$ and conjunctions of positive atoms in variables
$\{v_1,\ldots,v_m\}$. From a structure $\bB$ build the conjunction $\phi_\bB$
listing the tuples that hold on $\bB$ in which element $i$ corresponds to
variable $v_i$. Likewise, for a conjunction of positive atoms $\psi$, let
$\mathcal{D}_\psi$ be the structure whose relation tuples are listed by $\psi$,
where variable $v_i$ corresponds to element $i$. The relationship of $\bB$ to
$\phi_\bB$ and $\psi$ to $\mathcal{D}_\psi$ is very similar to that of
\emph{canonical query} and \emph{canonical database} (see
\cite{KolaitisVardiBook05}), except there we consider the conjunctions of atoms
to be existentially quantified. For example, $\mathcal{K}_3$ on domain
$\{1,2,3\}$ gives rise to $\phi_{\mathcal{K}_3}:= \exists v_1,v_2,v_3 \
E(v_1,v_2) \wedge E(v_2,v_1) \wedge E(v_2,v_3) \wedge$ $E(v_3,v_2) \wedge
E(v_3,v_1) \wedge E(v_3,v_1)$.  The Prover-Adversary game
$\mathscr{G}(\Psi,\bB)$ may be seen as Prover giving $j$ potential maps for
element $x$ in $\mathcal{D}_\psi$ ($\psi$ is quantifier-free part of $\Psi$) and
Adversary choosing one of them. The winning condition for Prover is now that the
map given from $\mathcal{D}_\psi$ to $\bB$ is a homomorphism.

In the case of QCSP, \mbox{i.e.} $\{1,|B|\}$-pp, each move of a game
$\mathscr{G}(\Psi,\bB)$ is trivial for one of the players. For $\exists^{\geq
1}$ quantifiers, Prover gives a singleton set, so Adversary's choice is forced.
In the case of $\exists^{\geq |B|}$ quantifiers, Prover must advance all of $B$.
Thus, essentially, Prover alone plays $\exists^{\geq 1}$ quantifiers and
Adversary alone plays $\exists^{\geq |B|}$ quantifiers.
\vspace{-1.5ex}

\section{Complexity of a single quantifier}\label{sec:single}\vspace{-1ex}

In this section we consider the complexity of evaluating $X$-pp sentences when
$X$ is a singleton, \mbox{i.e.}, we have at our disposal only a single
quantifier.\vspace{-0.5ex}

{\em
\begin{enumerate}[(1)]
\item $\{1\}$-CSP$(\bB)$ is in $\mathrm{NP}$ for all $\bB$. For each $n \geq 2$,
there exists a template $\bB_n$ of size $n$ such that $\{1\}$-CSP$(\bB_n)$ is
$\mathrm{NP}$-complete. 

\item $\{|B|\}$-CSP$(\bB)$ is in $\mathrm{L}$ for all $\bB$. 

\item For each $n \geq 3$, there exists a template $\bB_n$ of size $n$ such that
$\{j\}$-CSP$(\bB_n)$ is $\mathrm{Pspace}$-complete for all $1<j<n$.
\end{enumerate}
}\vspace{-0.5ex}

\begin{proof}
Parts (1) and (2) are well-known (see \cite{Papa}, resp. \cite{CiE2008}).  For
(3), let $\bB_{\mathrm{NAE}}$ be the Boolean structure on domain $\{0,1\}$ with
a single ternary not-all-equal relation $R_{\mathrm{NAE}}:=\{0,1\}^3 \setminus
\{(0,0,0),(1,1,1)\}$.  To show Pspace-completeness, we reduce from
QCSP$(\bB_{\mathrm{NAE}})$, the {\em quantified
not-all-equal-$3$-satisfiability} (see \cite{Papa}).

We distinguish two cases.

\noindent{\bf Case I:} $j \leq \lfloor n/2 \rfloor$. Define $\bB_{n}$ on domain
$\{0,\ldots,n-1\}$ with a single unary relation $U$ and a single ternary
relation $R$. Set $U:=\{0,\ldots,j-1\}$ and set 
\smallskip

$R:=\{0,\ldots,n-1\}^3
\setminus \{(a,b,c) : \mbox{ $a,b,c$ either all odd or all even}\}.$
\smallskip

\noindent The even numbers will play the role of false $0$ and odd numbers
the role~of~true~$1$.\smallskip

\noindent{\bf Case II:} $j > \lfloor n/2 \rfloor$. Define $\bB_{n}$ on domain
$\{0,\ldots,n-1\}$ with a single unary relation $U$ and a single ternary
relation $R$. Set $U:=\{0,\ldots,j-1\}$ and set\smallskip

\noindent\hfill$R:=\{0,\ldots,n-1\}^3
\setminus \{(a,b,c) : \mbox{ $a,b,c\leq n-j$ and either all odd or all
even}\}.$\hfill\mbox{}\smallskip

\noindent In this case even numbers $\leq n-j$ play the role of false $0$ and
odd numbers $\leq n-j$ play the role of true $1$. The $j-1$ numbers
$n-j+1,\ldots,n-1$ are somehow universal and will always satisfy any $R$
relation.

The reduction we use is the same for Cases I and II. We reduce
QCSP$(\bB_{\mathrm{NAE}})$ to $\{j\}$-CSP$(\bB_{n})$. Given an input $\Psi:=Q_1
x_1 Q_2 x_2 \ldots Q_m x_m \ \psi(x_1,x_2,$ $\ldots,x_m)$ to the former
(\mbox{i.e.} each $Q_i$ is $\exists$ or $\forall$) we build an instance $\Psi'$
for the latter.  From the outside in, we convert quantifiers $\exists x$ to
$\exists^{\geq j} x$.  For quantifiers $\forall x$, we convert also to
$\exists^{\geq j} x$, but we add the conjunct $U(x)$ to the quantifier-free part
$\psi$.

We claim $\bB_{\mathrm{NAE}} \models \Psi$ iff $\bB_{n} \models \Psi'$.  For the
$\exists$ variables of $\Psi$, we can see that any $j$ witnesses from the domain
$B_{n}$ for $\exists^{\geq j}$ must include some element playing the role of
either false $0$ or true $1$ (and the other $j-1$ may always be found
somewhere). For the $\forall$ variables of $\Psi$, $U$ forces us to choose both
$0$ and $1$ among the $\exists^{\geq j}$ (and the other $j-2$ will come from
$2,\ldots,j-1$).~The~result~follows.\end{proof}\vspace{-3ex}

\section{Counting quantifiers on cliques and cycles}
\label{sec:cliques_and_cycles}
\vspace{-0.5ex}

\subsection{Cliques: proof of Theorem~\ref{thm:cliques}}
\vspace{-1ex}
\label{sec:cliques}

Recall that $\bK_n$ is the complete irreflexive graph on $n$ vertices. 

\begin{figure}[h!]
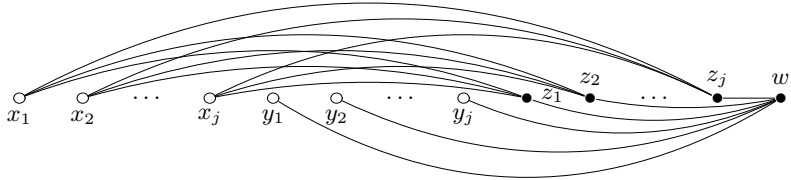

\vspace{-3ex}
\mbox{}\hfill
$\xy/r2pc/:
(0,0)*[o][F]{\phantom{s}}="x1";
(1,0)*[o][F]{\phantom{s}}="x2";
(2,0)*{\ldots};
(3,0)*[o][F]{\phantom{s}}="x4";
(4,0)*[o][F]{\phantom{s}}="y1";
(5,0)*[o][F]{\phantom{s}}="y2";
(6,0)*{\ldots};
(7,0)*[o][F]{\phantom{s}}="y4";
(8,0)*[o]{\bullet}="z1";
(9,0)*[o]{\bullet}="z2";
(10,0)*{\ldots};
(11,0)*[o]{\bullet}="z4";
(12,0)*[o]{\bullet}="w";
{\ar@{-}@/^1.5pc/ "x1";"z1"};
{\ar@{-}@/^2pc/ "x1";"z2"};
{\ar@{-}@/^3pc/ "x1";"z4"};
{\ar@{-}@/^1pc/ "x2";"z1"};
{\ar@{-}@/^1.5pc/ "x2";"z2"};
{\ar@{-}@/^2.5pc/ "x2";"z4"};
{\ar@{-}@/^0.5pc/ "x4";"z1"};
{\ar@{-}@/^1pc/ "x4";"z2"};
{\ar@{-}@/^2pc/ "x4";"z4"};
{\ar@{-}@/^2.5pc/ "w";"y1"};
{\ar@{-}@/^1.7pc/ "w";"y2"};
{\ar@{-}@/^1.1pc/ "w";"y4"};
{\ar@{-}@/^0.7pc/ "w";"z1"};
{\ar@{-}@/^0.3pc/ "w";"z2"};
{\ar@{-}@/^0pc/ "w";"z4"};
"x1"+(0,-0.3)*{x_1};
"x2"+(0,-0.3)*{x_2};
"x4"+(0,-0.3)*{x_j};
"y1"+(0,-0.3)*{y_1};
"y2"+(0,-0.3)*{y_2};
"y4"+(0,-0.3)*{y_j};
"z1"+(0.4,0.04)*{z_1};
"z2"+(0,0.3)*{z_2};
"z4"+(0,0.3)*{z_j};
"w"+(0,0.3)*{w};
\endxy
$\hfill\mbox{}
\vspace{-1ex}
\caption{The gadget $\mathcal{G}_j$.\label{fig:juraj}}
\vspace{-6ex}
\end{figure}

\begin{proposition}\label{prop:jCSP2j+1}
If $1<j$, then $\{j\}$-CSP$(\bK_{2j+1})$ is $\mathrm{Pspace}$-complete.
\end{proposition}\vspace{-1ex}

\begin{proof}
By reduction from QCSP$\big(\mathcal{K}_{\binom{2j+1}{j}}\big)$, \emph{quantified
$\binom{2j+1}{j}$-colouring}, which is Pspace-complete by
\cite{OxfordQuantifiedConstraints}. The key part of our proof involves the
gadget $\mathcal{G}_j$, in Figure \ref{fig:juraj}, having vertices
$x_1,\ldots,x_j,y_1,\ldots,y_j,z_1,\ldots,z_j,w$ and all possible edges between
$\{x_1,\ldots,x_j\}$ and $\{z_1,\ldots,z_j\}$, and between $w$ and
$\{y_1,\ldots,y_j,z_1,\ldots,z_j\}$.  The left $2j$ vertices represent free
variables $x_1,\ldots,x_j,$ $y_1,\ldots,y_j$.  Observe that $\exists^{\geq j}
z_1,\ldots,z_j,w \ \phi_{\mathcal{G}_j}$ is true on $\bK_{2j+1}$ iff
$|\{x_1,\ldots,x_j\} \cap \{y_1,\ldots,y_j\}|<j$.  If $|\{x_1,\ldots,x_j\}|=
|\{y_1,\ldots,y_j\}|=j$, this is equivalent~to $\{x_1\ldots x_j\} \neq
\{y_1\ldots y_j\}$. Thus this gadget will help us to encode the edge relation on
$\mathcal{K}_{\binom{2j+1}{j}}$ where we represent vertices by sets
$\{a_1,\ldots,a_j\} \subset \{1,\ldots,2j+1\}$ with $|\{a_1,\ldots,a_j\}|=j$.

Consider an instance $\Psi$ of QCSP$\big(\mathcal{K}_{\binom{2j+1}{j}}\big)$. We
construct the instance $\Psi'$ of $\{j\}$-CSP$(\bK_{2j+1})$ as follows. From the
graph $\mathcal{D}_\psi$, build $\mathcal{D}_{\psi'}$ by transforming each
vertex $v$ into an independent set of $j$ vertices $\{v^1,\ldots,v^j\}$, and
each edge $uv$ of ${\cal D}_{\psi}$ to an instance of the gadget $\mathcal{G}_j$
in which the $2j$ free variables correspond to $u^1,\ldots,u^j,v^1,\ldots,v^j$.
The other variables of the gadget $\{z_1,\ldots,z_j,w\}$ are unique to each edge
and are quantified innermost in $\Psi'$ in the order $z_1,\ldots,z_j,w$. 

It remains to explain the quantification of the variables of the form
$v^1,\ldots,v^j$. We follow the quantifier order of $\Psi$.  Existentially
quantified variables $\exists v$ of $\Psi$ are quantified as $\exists^{\geq j}
v^1,\ldots,v^j$ in $\Psi'$.  Universally quantified variables $\forall v$ of
$\Psi$ are also quantified $\exists^{\geq j} v^1,\ldots,v^j$ in $\Psi'$, but we
introduce additional variables $v^{1,1},\ldots,v^{1,j+1},$ $\ldots,$
$v^{j,1},\ldots,v^{j,j+1}$ before $v^1,\ldots,v^j$ in the quantifier order of
$\Psi'$, and for each $i \in \{1,\ldots,j\}$, we join $v^{i,1},\ldots,v^{i,j+1}$
into a clique with $v^i$.

It is now not difficult to verify that $\bK_{\binom{2j+1}{j}} \models \Psi$ iff
$\bK_{2j+1} \models \Psi'$.  \\(For lack of space, we omit further details;
please consult the appendix.)
\end{proof}
\vspace{-3ex}

\begin{corollary}\label{cor:jCSPn}
If $1<j<n/2$, then $\{j\}$-CSP$(\bK_{n})$ is $\mathrm{Pspace}$-complete.
\end{corollary}\vspace{-1ex}

\begin{proof}
We reduce from $\{j\}$-CSP$(\bK_{2j+1})$ and appeal to
Proposition~\ref{prop:jCSP2j+1}. Given an input $\Psi$ for
$\{j\}$-CSP$(\bK_{2j+1})$, we build an instance $\Psi'$ for
$\{j\}$-CSP$(\bK_{n})$ by adding an $(n-2j-1)$-clique on new variables,
quantified outermost in $\Psi'$, and link by an edge each variable of this
clique to every other variable. Adversary chooses $n-2j-1$ elements of the
domain for this clique, effectively reducing the domain size to $2j+1$ for the
rest. Thus $\bK_n \models \Psi'$ iff $\bK_{2j+1} \models \Psi$ follows.
\end{proof}
\vspace{-3ex}

\begin{proposition}\label{prop:ijCSPn}
If $1<j\leq n$, then $\{1,j\}$-CSP$(\bK_{n})$ is $\mathrm{Pspace}$-complete.
\end{proposition}\vspace{-1ex}

\begin{proof}
By reduction from QCSP$(\bK_n)$. We simulate existential quantification $\exists
v$ by itself, and universal quantification $\forall v$ by the
introduction of $(n-j+1)$ new variables $v^1,\ldots,v^{n-j}$, joined in a clique
with $v$, and quantified by $\exists^{\geq j}$ before $v$ (which
is also quantified by $\exists^{\geq j}$). The argument follows as in
Proposition~\ref{prop:jCSP2j+1}.
\end{proof}
\vspace{-1ex}

Define the \mbox{$n$-star} $\bK_{1,n}$ to be the graph on vertices
$\{0,1,\ldots,n\}$ with edges \mbox{$\{(0,j),(j,0)~:~j\geq 1\}$} where $0$ is
called the \emph{centre} and the remainder are \emph{leaves}. 

\begin{proposition}\label{prop:easy-cliques}
If $X\cap\{1,\ldots,\lfloor n/2 \rfloor\} = \emptyset$, then $X$-CSP$(\bK_n)$ is
in $\mathrm{L}$.
\end{proposition}\vspace{-1ex}

\begin{proof}
Instance $\Psi$ of $X$-CSP$(\bK_n)$ of the form $\exists^{\geq \lambda_1} x_1
\ldots \exists^{\geq \lambda_m} x_m $ $\psi(x_1,\ldots,x_m)$ induces the graph
$\mathcal{D}_\psi$, which we may consider totally ordered (the order is given
left-to-right ascending by the quantifiers). We claim that $\bK_n \models \Psi$
iff $\mathcal{D}_\psi$ does not contain as a subgraph (not necessarily induced)
a $(n-\lambda_i+1)$-star in which the $n-\lambda_i+1$ leaves all come before the
centre $x_i$ in the ordering.

($\Rightarrow$) If $\mathcal{D}_\psi$ contains such a star, then $\Psi$
is a no-instance, as we may give a winning strategy for Adversary in the game
$\mathscr{G}(\Psi,\bK_n)$. Adversary should choose distinct values for the
variables associated with the $n-\lambda_i+1$ leaves of the star (can always be
done as each of the possible quantifiers assert existence of $> n/2$ elements
and $n-\lambda_i<n/2$), whereupon there is no possibility for Prover to
choose $\lambda_i$ witnesses to the variable $x_i$ associated with the centre.

($\Leftarrow$) If $\mathcal{D}_\psi$ does not contain such a star, then we give the
following winning strategy for Prover in the game $\mathscr{G}(\Psi,\bK_n)$.
Whenever a new variable comes up, its corresponding vertex in $\mathcal{D}_\psi$
has $l<n-\lambda_i+1$ adjacent predecessors, which were answered with
$b_1,\ldots,b_l$. Prover suggests any set of size $\lambda_i$ from $B \setminus
\{b_1,\ldots,b_l\}$ (which always exists) and the result follows.
\end{proof}
\vspace{-2ex}

\begin{proofof}{Theorem~\ref{thm:cliques}}
For $n\leq 2$ see \cite{CiE2006}, and for (ii) see \cite{HellNesetril}.  The
remainder of (i) is proved as Proposition~\ref{prop:easy-cliques} while
Corollary~\ref{cor:jCSPn} and Proposition \ref{prop:ijCSPn} give
(iii).~\end{proofof}
\vspace{-4.5ex}

\subsection{Cycles: proof of Theorem~\ref{thm:cycles}}
\vspace{-0.5ex}

Recall that $\bC_n$ denotes the irreflexive symmetric cycle on $n$ vertices.
We consider $\bC_n$ to have vertices $\{0,1,\ldots,n-1\}$ and edges
$\big\{(i,j)~:~|i-j|\in\{1,n-1\}\big\}$.

In the forthcoming proof, we use the following elementary observation from
additive combinatorics.  Let $n\geq 2$, $j\geq 1$, and $A,B$ be sets of
integers. Define:\vspace{0.3ex}

$\bullet~~A+_nB=\big\{(a+b)~{\rm mod}~n~\big|~a\in A,\,b\in B\big\}\qquad\bullet~~ j\times_n
A=\underbrace{A+_n\ldots+_n A}_{j~{\rm times}}$
\vspace{-4ex}

\begin{lemma}\label{lem:additive}
Let $n\geq 3$ and $2\leq j<n$. Then

\noindent\hfill
$\Big|\,j\times_n\{-1,+1\}\,\Big|=\left\{\begin{array}{l@{\quad}l}j+1&n~{\rm
is~odd}\\ \min\big\{j+1,n/2\big\}&n~{\rm is~even}\end{array}\right.$\hfill\mbox{}

~~\,$\Big|\,n\times_n\{-1,+1\}\,\Big|=
\Big|\,n\times_n\{-2,0,+2\}\,\Big|
=\left\{\begin{array}{l@{\quad}l}n&n~{\rm
is~odd}\\ n/2&n~{\rm is~even}\end{array}\right.$\hfill\mbox{}
\end{lemma}\vspace{-3ex}

\begin{proposition}\label{prop:cycles-i}
If $n\geq 3$, then $X$-CSP$(\bC_n)$ is in {\rm L} if $n=4$, or $1\not\in X$, or
$n$ is even and $X\cap\{2,3\ldots, n/2\}=\emptyset$,
\end{proposition}\vspace{-1ex}

\begin{proof}
Let $\Psi$ be an instance of $X$-CSP($\bC_n$).  Recall that ${\cal D}_\psi$ is
the graph corresponding to the quantifier-free part of $\Psi$.  We write $x\prec
y$ if $x,y$ are vertices of ${\cal D}_\psi$ (i.e., variables of $\psi$) such
that $x$ is quantified before $y$ in $\Psi$.  For an edge $xy$ of ${\cal
D}_\psi$ where $x\prec y$, we say that $x$ is a {\em predecessor} of $y$.  Note
that a vertex can have several predecessors.

The following claims restrict the yes-instances of $X$-CSP($\bC_n$).\vspace{1ex}

\noindent{\em Let $x$ be a vertex of ${\cal D}_\psi$ quantified in $\Psi$ by
$\exists^{\geq j}$ for some $j$. If $\bC_n \models \Psi$ then
\vspace{-1.5ex}
\begin{enumerate}[({1}a)]
\item if $j\geq 3$, then $x$ has no predecessors,
\item if $n$ is even and $j>n/2$, then $x$ is the first vertex (w.r.t.
$\prec$) of some connected component of ${\cal D}_\psi$, and
\item if $n\neq 4$ and $j=2$, then all predecessors of $x$
except for its first
predecessor (w.r.t. $\prec$) are quantified by~$\exists^{\geq 1}$.
\end{enumerate}\vspace{-1.5ex}
}

\noindent (We omit the proof for lack of space; please consult the
appendix.)\vspace{0.5ex}

Using these claims, we prove the proposition. First, we consider the case
\mbox{$n=4$.} We show that $\{1,2,3,4\}$-CSP($\bC_4$) is in L. This will imply
that $X$-CSP($\bC_4$) is in L for every $X$. Observe that if ${\cal D}_\psi$
contains a vertex $x$ quantified by $\exists^{\geq 3}$ or $\exists^{\geq 4}$,
then by (1b) this vertex is the first in its component (if $\Psi$ is not a
trivial no-instance). Thus replacing its quantification by $\exists^{\geq 1}$
does not change the truth of $\Psi$. So we may assume that $\Psi$ is an instance
of $\{1,2\}$-CSP($\bC_4$).  We now claim that $\bC_4\models\Psi$ if and only if
${\cal D}_\psi$ is bipartite.  Clearly, if ${\cal D}_\psi$ is not bipartite, it
has no homomorphism to $\bC_4$ and hence $\bC_4\not\models\Psi$.  Conversely,
assume that ${\cal D}_\psi$ is bipartite with bipartition $(A,B)$. Our strategy
for Prover offers the set $\{0,2\}$ or its subsets for the vertices in $A$ and
offers $\{1,3\}$ or its subsets for every vertex in $B$. It is easy to verify
that this is a winning strategy for Prover. Thus $\bC_4\models\Psi$.  The
complexity now follows as checking (1b) and checking
if~a~graph~is~bipartite~is~in~L~by~\cite{RheingoldJACM}.

Now, we may assume $n\neq 4$, and next we consider the case $1\not\in X$.  If
also $2\not\in X$, then by (1a) the graph ${\cal D}_\psi$ contains no edges
(otherwise $\Psi$ is a trivial no-instance).  This is clearly easy to check in
L. Thus $2\in X$. We claim that if we satisfy (1a) and (1c), then
$\bC_n\models\Psi$.  We provide a winning strategy for Prover.  Namely, for a
vertex $x$, if $x$ has no predecessors, offer any set for $x$. If $x$ has a
unique predecessor $y$ for which the value $i$ was chosen, then $x$ is
quantified by $\exists^{\geq 2}$ (or $\exists$) by (1a) and we offer
$\{i-1,i+1\}$ (mod $n$) for $x$ . There are no other cases by (1a) and (1c).  It
follows that Prover always wins with this strategy.  In terms of complexity, it
suffices to check (1a) and~(1c)~which~is~in~L.

Finally, suppose that $n$ is even and $X\cap\{2\ldots n/2\}=\emptyset$.  Note
that every vertex of ${\cal D}_\psi$ is either quantified by $\exists^{\geq 1}$
or by $\exists^{\geq j}$ where $j>n/2$. Thus, using (1b), unless $\Psi$ is a
trivial no-instance, we can again replace every $\exists^{\geq j}$ in $\Psi$ by
$\exists^{\geq 1}$ without changing the truth of $\Psi$. Hence, we may assume
that $\Psi$ is an instance of $\{1\}$-CSP($\bC_n$).  Thus, as $n$ is even,
$\bC_n\models \Psi$ if and only if ${\cal D}_\psi$ is bipartite. The complexity
again follows from \cite{RheingoldJACM}. That concludes the proof.
\end{proof}
\vspace{-3ex}

\begin{proposition}\label{prop:cycles-iii}
Let $n\geq 3$. Then $X$-CSP$(\bC_n)$ is Pspace-complete if $n\neq 4$ and
$\{1,j\}\subseteq X$: where $j \in \{2,\ldots,n\}$ if $n$ is odd and $j \in
\{2,\ldots,n/2\}$ if $n$ is even.
\end{proposition}\vspace{-1ex}

\begin{proof}
By reduction, namely a reduction from QCSP($\bC_n$) for odd $n$, and from
QCSP($\bK_{n/2}$) for even $n$. Both problems are known to be Pspace-hard
\cite{OxfordQuantifiedConstraints}.

First, consider the case of odd $n$. Let $\Psi$ be an instance of QCSP($\bC_n$).
In other words, $\Psi$ is an instance of $\{1,n\}$-CSP($\bC_n$).  Clearly, $j<n$
otherwise~we~are~done. 

We modify $\Psi$ by replacing each universally-quantified variable $x$ of $\Psi$
by a path. Namely, let $\pi_x$ denote the pp-formula that encodes
that\smallskip

\qquad $x^1_1,x^1_2,\ldots,x^1_{j-1},~x^2_1,x^2_2,\ldots,x^2_{j-1},\quad
\ldots\quad,x^{n}_1,x^{n}_2,\ldots,x^{n}_{j-1},~x$
\smallskip

\noindent is a path in that order (all but $x$ are new variables).
We replace $\forall x$  by \smallskip

\noindent\hfill$Q_x~=~\exists^{\geq j}x^1_1~\exists^{\geq j}x^2_1\ldots~\exists^{\geq j}x^{n}_1
~\exists^{\geq j}x~\exists^{\geq 1} x^1_2\ldots\exists^{\geq 1}
x^1_{j-1}~\ldots~\exists^{\geq 1} x^{n}_{2}\ldots\exists^{\geq 1}
x^{n}_{j-1}$\hfill\mbox{}
\smallskip

\noindent and append $\pi_x$ to the quantifier-free part of the formula. Let
$\Psi'$ denote the final formula after considering all universally quantified
variables. Note that $\Psi'$ is an instance of $\{1,j\}$-CSP($\bC_n$).  We claim
that $\bC_n\models\Psi$ if and only if $\bC_n\models\Psi'$.

To do this, it suffices to show that $\Psi'$ correctly simulates the universal
quantifiers of $\Psi$.  Namely, we prove that $\bC_n\models Q_x\pi_x$,
and for each $\ell\in\{0\ldots n-1\}$, Adversary has a strategy on
$Q_x\pi_x$ that evaluates $x$ to $\ell$.\vspace{0.5ex}\\
(We omit further details for lack of space; please consult the appendix.)
\medskip

It remains to investigate the case of even $n$.  Recall that $n\geq 6$ and
$j\leq n/2$. We show a reduction from QCSP($\bK_{n/2}$) to $\{1,j\}$-QCSP($\bC_n$).
The reduction is a variant of the construction from \cite{FederHellHuang99} for
the problem of retraction to even cycles.

Let $\Psi$ be an instance of QCSP($\bK_{n/2}$), and define $r=(-n/2-2)~{\rm
mod~}(j-1)$.  We construct a formula $\Psi'$ from $\Psi$ as follows.  First, we
modify $\Psi$ by replacing universal quantifiers exactly as in the case of odd
$n$.  Namely, we define $Q_x$ and $\pi_x$ as before, replace each $\forall x$ by
$Q_x$, and append $\pi_x$ to the quantifier-free part of the formula.  After
this, we append to the formula a cycle on $n$ vertices $v_0,v_1,\ldots,v_{n-1}$
with a path on $r+1$ vertices $w_0,w_1,\ldots,w_r$.  (See the black vertices in
Figure \ref{fig:even-case}.) Then, for each edge $xy$ of ${\cal D}_{\psi}$, we
replace $E(x,y)$ in $\Psi$ by the gadget depicted in Figure~\ref{fig:even-case}
(consisting of the cartesian product of $\bC_n$ and a path on $3n/2$ vertices
together with two attached paths on $n/2-2$, resp. $r+1$ vertices).  The
vertices $x$ and $y$ represent the variables $x$ and $y$ while all other white
vertices are new variables, and the black vertices are identified with
$v_0,\ldots,v_{n-1},w_0,\ldots,w_r$ introduced in the previous step.

Finally, we prepend the following quantification to the formula:
\smallskip

\hfill$\exists^{\geq 1} w_0\,\exists^{\geq j} v_{j-r-2}\,\exists^{\geq j}
v_{2j-r-3}\ldots \exists^{\geq j} v_{(k\cdot j-r-k-1)} \ldots \exists^{\geq j}
v_{n/2+1}$\hfill\mbox{}
\smallskip

\noindent followed by $\exists^{\geq 1}$ quantification of all the remaining
variables of the gadgets.

\begin{figure}[t!]
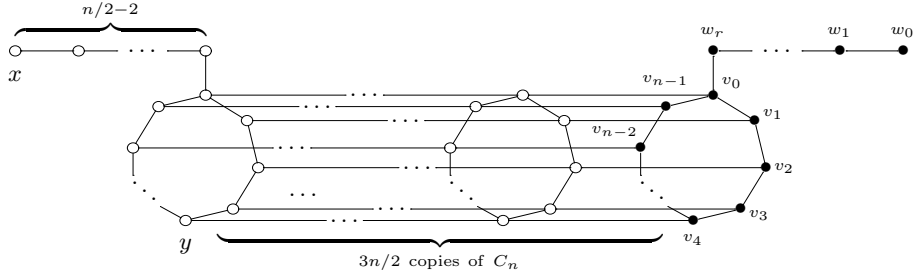

\mbox{}\hfill$
\xy/r2pc/:
(9,0);(9,0);{\xypolygon8"X"{~>{}~={81}[o]{\phantom{s}}}};
(6,0);(6,0);{\xypolygon8"Y"{~>{}~={81}[o]{\phantom{s}}}};
(1,0);(0,0);{\xypolygon8"Z"{~>{}~={81}[o]{\phantom{s}}}};
"X1"*[o]{\bullet};
"X2"*[o]{\bullet};
"X3"*[o]{\bullet};
"X4"+(0,0.2)*{\ddots};
"X5"*[o]{\bullet};
"X6"*[o]{\bullet};
"X7"*[o]{\bullet};
"X8"*[o]{\bullet};
"Y1"*[o][F]{\phantom{s}};
"Y2"*[o][F]{\phantom{s}};
"Y3"*[o][F]{\phantom{s}};
"Y4"+(0,0.2)*{\ddots};
"Y5"*[o][F]{\phantom{s}};
"Y6"*[o][F]{\phantom{s}};
"Y7"*[o][F]{\phantom{s}};
"Y8"*[o][F]{\phantom{s}};
"Z1"*[o][F]{\phantom{s}};
"Z2"*[o][F]{\phantom{s}};
"Z3"*[o][F]{\phantom{s}};
"Z4"+(0,0.2)*{\ddots};
"Z5"*[o][F]{\phantom{s}};
"Z6"*[o][F]{\phantom{s}};
"Z7"*[o][F]{\phantom{s}};
"Z8"*[o][F]{\phantom{s}};
{\ar@{-} "X1";"X2"};
{\ar@{-} "X2";"X3"};
{\ar@{-} "X3";"X3"+(0,-0.43)};
{\ar@{-} "X5";"X5"+(-0.4,0.22)};
{\ar@{-} "X5";"X6"};
{\ar@{-} "X6";"X7"};
{\ar@{-} "X7";"X8"};
{\ar@{-} "X8";"X1"};
{\ar@{-} "Y1";"Y2"};
{\ar@{-} "Y2";"Y3"};
{\ar@{-} "Y3";"Y3"+(0,-0.43)};
{\ar@{-} "Y5";"Y5"+(-0.4,0.22)};
{\ar@{-} "Y5";"Y6"};
{\ar@{-} "Y6";"Y7"};
{\ar@{-} "Y7";"Y8"};
{\ar@{-} "Z1";"Z2"};
{\ar@{-} "Z2";"Z3"};
{\ar@{-} "Z3";"Z3"+(0,-0.43)};
{\ar@{-} "Z5";"Z5"+(-0.4,0.22)};
{\ar@{-} "Z5";"Z6"};
{\ar@{-} "Z6";"Z7"};
{\ar@{-} "Z7";"Z8"};
{\ar@{-} "Z8";"Z1"};
{\ar@{-} "Y8";"Y1"};
{\ar@{-} "X1";"Y1"};
{\ar@{-} "X2";"Y2"};
{\ar@{-} "X3";"Y3"};
{\ar@{-} "X5";"Y5"};
{\ar@{-} "X6";"Y6"};
{\ar@{-} "X7";"Y7"};
{\ar@{-} "X8";"Y8"};
{\ar@{-} "Y1";"Y1"+(-2.2,0)};
{\ar@{-} "Y2";"Y2"+(-2.2,0)};
{\ar@{-} "Y3";"Y3"+(-2.2,0)};
{\ar@{-} "Y5";"Y5"+(-2.2,0)};
{\ar@{-} "Y6";"Y6"+(-2.2,0)};
{\ar@{-} "Y7";"Y7"+(-2.2,0)};
{\ar@{-} "Y8";"Y8"+(-2.2,0)};
{\ar@{-} "Z1";"Z1"+(2.2,0)};
{\ar@{-} "Z2";"Z2"+(2.2,0)};
{\ar@{-} "Z3";"Z3"+(2.2,0)};
{\ar@{-} "Z5";"Z5"+(2.2,0)};
{\ar@{-} "Z6";"Z6"+(2.2,0)};
{\ar@{-} "Z7";"Z7"+(2.2,0)};
{\ar@{-} "Z8";"Z8"+(2.2,0)};
"Y1"+(-2.5,0)*{\ldots};
"Y2"+(-2.5,0)*{\ldots};
"Y3"+(-2.5,0)*{\ldots};
"Y4"+(-2.5,0)*{\ldots};
"Y5"+(-2.5,0)*{\ldots};
"Y6"+(-2.5,0)*{\ldots};
"Y7"+(-2.5,0)*{\ldots};
"Y8"+(-2.5,0)*{\ldots};
"Z1"+(-0,0.7)*[o][F]{\phantom{s}}="a1";
"a1"+(-1,0)*{\ldots}="a2";
"a2"+(-1,0)*[o][F]{\phantom{s}}="a3";
"a3"+(-1,0)*[o][F]{\phantom{s}}="a4";
{\ar@{-} "Z1";"a1"};
{\ar@{-} "a1";"a1"+(-0.6,0)};
{\ar@{-} "a3";"a3"+(0.6,0)};
{\ar@{-} "a3";"a4"};
"X1"+(0,0.7)*{\bullet}="b1";
"b1"+(1,0)*{\ldots}="b2";
"b2"+(1,0)*{\bullet}="b3";
"b3"+(1,0)*{\bullet}="b4";
{\ar@{-} "X1";"b1"};
{\ar@{-} "b1";"b1"+(0.6,0)};
{\ar@{-} "b3";"b3"+(-0.6,0)};
{\ar@{-} "b3";"b4"};
"a4"+(0,-0.4)*{x};
"Z5"+(0,-0.4)*{y};
"X1"+(0.3,0.2)*{_{v_0}};
"X8"+(0.3,0.1)*{_{v_1}};
"X7"+(0.3,0)*{_{v_2}};
"X6"+(0.3,-0.1)*{_{v_3}};
"X5"+(0,-0.3)*{_{v_4}};
"X3"+(-0.4,0.25)*{_{v_{n-2}}};
"X2"+(0,0.4)*{_{v_{n-1}}};
"b1"+(0,0.3)*{_{w_r}};
"b3"+(0,0.3)*{_{w_1}};
"b4"+(0,0.3)*{_{w_0}};
"a1"+(-1.5,0.5)*{\overbrace{\hspace{6pc}}^{n/2-2}};
"Y5"+(-1,-0.5)*{\underbrace{\hspace{14pc}}_{3n/2~{\rm copies~of~}C_n}};
\endxy$\hfill\mbox{}
\vspace{-2ex}
\caption{The gadget for the case of even $n$ where $r=(-n/2-2)~{\rm
mod}~(j-1)$.\label{fig:even-case}}
\vspace{-3ex}
\end{figure}

We prove that $\bK_{n/2}\models\Psi$ if and only if $\bC_n\models\Psi'$.  First, we
show that $\Psi'$ correctly simulates the universal quantification of $\Psi$.
The argument for this is essentially the same as in the case of odd $n$.
Next, we need to analyse possible assignments to the vertices $v_0,\ldots,v_{n-1}$.
There are two possibilities:  either the values chosen for $v_0, \ldots,
v_{n-1}$ are all distinct, or not. In the former case, we show that Prover can
complete the homomorphism to $\bC_n$ if and only $\bK_{n/2}\models\Psi$.  All
other cases are degenerate and we address them separately.\vspace{0.5ex}\\
\noindent (We skip further details for lack of space; please consult the
appendix.) \end{proof}\vspace{-1ex}

\begin{proofof}{Theorem~\ref{thm:cycles}}
The case (i) is proved as Proposition \ref{prop:cycles-i}, and the case (ii)
follows from \cite{HellNesetril}. Finally, the case (iii) is proved as
Proposition \ref{prop:cycles-iii}.
\end{proofof}\vspace{-3ex}

\section{Extensions of the CSP}\label{sec:extensions}
\vspace{-0.5ex}

In this section we consider single-quantifier extensions of the classical
CSP$(\bB)$, \mbox{i.e.}, the evaluation of $X$-pp sentences, where $X:=\{1,j\}$
for some $1<j\leq |B|$.\vspace{-2ex}

\subsection{Bipartite graphs}\vspace{-1ex}

In the case of (irreflexive, undirected) graphs, it is known that
$\{1\}$-CSP$(\bH)=\mbox{CSP}(\bH)$ is in L if $\bH$ is bipartite and is
NP-complete otherwise \cite{HellNesetril} (for membership in L, one needs also
\cite{RheingoldJACM}). It is also known that something similar holds for
$\{1,|H|\}$-CSP$(\bH)=\mbox{QCSP}(\bH)$ -- this problem is in L if $\bH$ is
bipartite and is NP-hard otherwise \cite{CiE2006}. Of course, the fact that
$\{1,j\}$-CSP$(\bH)$ is hard on non-bipartite $\bH$ is clear, but we will see
that it is not always easy on~bipartite~$\bH$.

First, we look at complete bipartite graphs (in a more general statement).
\vspace{-1ex}

\begin{proposition}
\label{prop:complete-bipartite}
Let $\bK_{k,l}$ be the complete bipartite graph with partite sets of size $k$
and $l$. Then $\{1,\ldots,k+l\}$-CSP$(\bK_{k,l})$ is in $\mathrm{L}$.
\end{proposition}

\begin{proof}
We reduce to QCSP$(\bK^1_2)$, where $\bK^1_2$ indicates $\bK_2$ with one vertex
named by a constant, say $1$. QCSP$(\bK^1_2)$ is equivalent to QCSP$(\bK_2)$
(identify instances if $1$ to a single vertex) and both are well-known to be in
L (see, e.g., \cite{CiE2006}). Let $\Psi$ be input to
$\{1,\ldots,k+l\}$-CSP$(\bK_{k,l})$. Produce $\Psi'$ by substituting quantifiers
$\exists^{\geq j}$ with $\exists$, if $j \leq \mbox{min}\{k,l\}$, or with
$\forall$, if $j > \mbox{max}\{k,l\}$. Variables quantified by $\exists^{\geq
j}$ for $\mbox{min}\{k,l\} < j \leq \mbox{max}\{k,l\}$ should be replaced by the
constant $1$.  It is easy to see that $\bK_{k,l} \models \Psi$ iff $\bK_2
\models \Psi'$, and the result follows.
\end{proof}\vspace{-3.2ex}

\begin{proposition}\label{prop:12C2j-extra}
For each $j$, there exists $m$ \mbox{s.t.} $[2^m1^*]$-CSP$(\bC_{2j})$ is
$\mathrm{NP}$-complete.
\end{proposition}\vspace{-1ex}

\begin{proof}
Membership in NP follows because $m$ is bounded -- one may try all possible
evaluations to the $\exists^{\geq 2}$ variables. NP-hardness follows as in the
proof of case $(iii)$ of Theorem~\ref{thm:cycles}, but we are reducing from
CSP$(\bK_{n/2})$ not QCSP$(\bK_{n/2})$. As a consequence, the only instances of
$\exists^{\geq 2}$ we need to consider are those used to isolate the cycle
$\bC_{2j}$ (one may take $m:=j+3$).
\end{proof}\vspace{-3.2ex}

\begin{corollary}\label{cor:twins}
If $\bH$ is bipartite and (for $j \geq 3$) contains some $\bC_{2j}$ but no
smaller cycle, then exists $m$ \mbox{s.t.} $[2^{m}1^*]$-CSP$(\bH)$ is
$\mathrm{NP}$-complete.
\end{corollary}\vspace{-1ex}

\begin{proof}
Membership and reductions for hardness follow similarly to
Proposition~\ref{prop:12C2j-extra}. The key part is in isolating a copy of the
cycle, but we can not do this as easily as before. If $d$ is the diameter of
$\bH$ (the maximum of the minimal distances between two vertices) then we begin
the sentence $\Psi'$ of the reduction with $\exists^{\geq 2}
v_1,\ldots,v_{d+1}$, and then, for each $i \in \{1,\ldots,d-j+1\}$ we add
$\exists^{\geq 2} x_i,x'_i, \ldots, x^{\prime\ldots\prime}$ ($j-1$ dashes) and join
$v_i,\ldots,v_{i+j},x^{\prime\ldots\prime}_i,\ldots,x_i$ in a $2i$-cycle (with
$E(x_i,v_i)$ also). For each of these $d-j+1$ cycles $\bC_{2j}$ we build a
separate copy of the rest of the reduction. We can not be sure which of these
cycles is evaluated truly on some $\bC_{2j}$, but at least one of them must be.
\end{proof}\vspace{-2.2ex}

\noindent (See the appendix for an example of why the above construction is
necessary.)
\vspace{0.2ex}

In passing, we note the following simple propositions.
\vspace{-1.2ex}

\begin{proposition}
\label{prop:2ton-3}
If $j \in \{2,...,n-3\}$ then one may exhibit a bipartite $\bH_j$ of size $n$
such that $\{1,j\}$-CSP$(\bH_j)$ is $\mathrm{Pspace}$-complete.
\end{proposition}
\vspace{-1ex}

\begin{proof}
The case $j=2$ follows from Theorem~\ref{thm:cycles}; assume $j\geq 3$. Take the
graph $\mathcal{C}_6$ and construct $\bH_j$ as follows. Augment $\mathcal{C}_6$
with $j-3$ independent vertices each with an edge to vertices $1$, $3$ and $5$
of $\mathcal{C}_6$. Apply the proof of
Theorem~\ref{thm:cycles}~with~$\bH_j$.~\end{proof}
\vspace{-3ex}

\begin{proposition}
\label{prop:bipartition}
Let $\bH$ be bipartite with largest partition in a connected component of size
$<j$. Then $\{1,j\}$-CSP$(\bH)$ is in $\mathrm{L}$.
\end{proposition}
\vspace{-1ex}

\begin{proof}
We will consider an input $\Psi$ to $\{1,j\}$-CSP$(\bH)$ of the form $Q_1 x_1 $
$Q_2 x_2 \ldots $ $Q_m x_m \ \psi(x_1,x_2,\ldots,x_m)$. An instance of an
$\exists^{\geq j}$ variable is called \emph{trivial} if it has neither a path to
another (distinct) $\exists^{\geq j}$ variable, nor a path to an $\exists$
variable that precedes it in the natural order on $\mathcal{D}_\psi$.  The key
observation here is that any non-trivial $\exists^{\geq j}$ variable \emph{must}
be evaluated on more than one partition of a connected component. If in $\Psi$
there is a non-trivial $\exists^{\geq j}$ variable, then $\Psi$ must be a
no-instance (as $\exists^{\geq j}$s must be evaluated on more than one partition
of a connected component, and a path can not be both even and odd in length).
All other instances are readily seen to be satisfiable. Detecting if $\Psi$
contains a non-trivial $\exists^{\geq j}$ variable is in L by
\cite{RheingoldJACM}, and the result follows.
\end{proof}

We note that Proposition~\ref{prop:2ton-3} is tight, namely in that
$\{1,j\}$-CSP$(\bH)$ is in L if $j \in \{1,|H|-2,|H|-1,|H|\}$. (For the proof,
please consult the appendix.)
\vspace{-0.5ex}

\begin{proposition}
\label{prop:cont-c4}
If $\bH$ is bipartite and contains $\bC_4$, then $\Psi \in
\{1,2\}$-CSP$(\bC_{4})$ iff the underlying graph $\mathcal{D}_\psi$ of $\Psi$ is
bipartite. In particular,  $\{1,2\}$-CSP$(\bH)$ is in L.
\end{proposition}\vspace{-1ex}

\begin{proof}
Necessity is clear; sufficiency follows by the canonical
evaluation of $\exists^{\geq 1}$ and $\exists^{\geq 2}$ on a fixed copy of
$\bC_4$ in $\bH$. Membership in L follows from
\cite{RheingoldJACM}.~\end{proof}\vspace{-3ex}

\begin{proposition}
\label{prop:12forest-extra}
Let $\bH$ be a forest, then $[2^{m}1^*]$-CSP$(\bH)$ is in $\mathrm{P}$ for all $m$.
\end{proposition}\vspace{-1ex}

\begin{proof}
We evaluate each of the $m$ variables bound by $\exists^{\geq 2}$ to all
possible pairs, and what we obtain in each case is an instance of CSP$(\bH')$
where $\bH'$ is an expansion of $\bH$ by some constants, \mbox{i.e.}, 
equivalent to the retraction problem. It is known that Ret$(\bH)$ is in P for
all forests $\bH$ \cite{Pseudoforests}, and the result follows. 
\end{proof}
\vspace{-2ex}

We bring together some previous results into a classification theorem.
\vspace{-1.5ex}

\begin{theorem} \label{thm:selling-point}
Let $\bH$ be a graph. Then\vspace{-1.5ex}
\begin{enumerate}[--]
\item $[2^{m}1^*]$-CSP$(\bH)\in{\rm P}$ for all $m$, if $\bH$ is a forest or a
bipartite graph~containing~$\bC_4$
\item $[2^{m}1^*]$-CSP$(\bH)$ is NP-complete from some $m$, if otherwise.
\end{enumerate}
\vspace{-2ex}
\end{theorem}

\begin{proof}
Membership of NP follows since $m$ is fixed. The cases in P follow from
Propositions~\ref{prop:12forest-extra} and \ref{prop:cont-c4}.  Hardness for
non-bipartite graphs follows from \cite{HellNesetril} and for the remaining
bipartite graphs it follows from Corollary~\ref{cor:twins}. 
\end{proof}
\vspace{-4ex}

\section{The complexity of QCSP$(\bC^*_4)$}\label{sec:new}
\vspace{-1ex}

Let $\bC^*_4$ be the reflexive $4$-cycle. The complexities of Ret$(\bC_6)$ and
Ret$(\bC^*_4)$ are both hard (NP-complete) \cite{FederHellHuang99,FederHell98},
and retraction is recognised to be a ``cousin'' of QCSP (see
\cite{SurHomSurvey}). The problem QCSP$(\bC_6)$ is known to be in L (see
\cite{CiE2006}), but the complexity of QCSP$(\bC^*_4)$ was hitherto unknown.
Perhaps surprisingly, we show that is is markedly different from that of
QCSP$(\bC_6)$, being Pspace-complete.
\vspace{-1ex}

\begin{proposition}\label{prop:c4*}
\label{prop:C1111}
$\{1,2,3,4\}$-CSP$(\bC^*_4)$ is $\mathrm{Pspace}$-complete.
\end{proposition}
\vspace{-2ex}

\begin{corollary}\label{cor:c4*}
QCSP$(\bC^*_4)$ is $\mathrm{Pspace}$-complete.
\end{corollary}
\vspace{-1ex}

\noindent The proofs of these claims are based on the hardness of the retraction
problem to reflexive cycles \cite{FederHell98} and are similar to our proof of
the even case of Proposition~\ref{prop:cycles-iii}.\\ (We skip the details for
lack of space; please consult the appendix.)

While QCSP$(\bC^*_4)$ has different complexity from QCSP$(\bC_6)$, we
remark that the better analog of the retraction complexities is perhaps that
$\{1,|C^*_4|\}$-CSP$(\bC^*_4)$ and $\{1,|C_6|/2\}$-CSP$(\bC_6)$ \textbf{do} have
the same complexities (recall the reductions to Ret$(\bC^*_4)$ and Ret$(\bC_6)$
involved CSP$(\mathcal{K}_{|C^*_4|})$ and CSP$(\mathcal{K}_{|C_6|/2})$,
respectively.
\vspace{-1ex}

\section{Conclusion}\label{sec:conclusion}
\vspace{-1ex}
We have taken first important steps to understanding the complexity of CSPs with
counting quantifiers, even though several interesting questions have resisted
solution.  We would like to close the paper with some open problems.
\newpage

In Section~\ref{sec:cliques}, the case $n=2j$ remains. When $j=1$ and $n=2$, we
have $\{1\}$-CSP$(\bK_2)$=CSP$(\bK_2)$ which is in L by \cite{RheingoldJACM}.
For higher $j$, the question of the complexity of $\{j\}$-CSP$(\bK_{2j})$ is
both challenging and open.

We would like to prove the following more natural variants of
Theorem~\ref{thm:selling-point}, whose involved combinatorics appear to be much
harder.

\begin{conjecture}
Let $\bH$ be a graph. Then 
\vspace{-0.5ex}
\begin{itemize}
\item $[2^*1^*]$-CSP$(\bH)$ is in P, if $\bH$ is a forest or a bipartite
graph containing $\bC_4$,
\item $[2^*1^*]$-CSP$(\bH)$ is NP-hard, if otherwise.
\end{itemize}
\end{conjecture}
\vspace{-1ex}

\begin{conjecture}
Let $\bH$ be a graph. Then 
\vspace{-0.5ex}
\begin{itemize}
\item $\{1,2\}$-CSP$(\bH)$ is in P, if $\bH$ is a forest or a bipartite
graph containing $\bC_4$,
\item $\{1,2\}$-CSP$(\bH)$ is NP-hard, if otherwise.
\end{itemize}
\end{conjecture}

\newpage

\appendix

\section{Appendix}

\subsection{Omitted proofs}

\begin{proofof}{Proposition \ref{prop:jCSP2j+1}}
We show that $\bK_{\binom{2j+1}{j}} \models \Psi$ iff $\bK_{2j+1} \models
\Psi'$.  Observe there is a natural bijection $\pi$ from subsets of $j$ elements
of $\bK_n$ to vertices of $\bK_{\binom{2j+1}{j}}$. In the simulation of
QCSP$(\bK_{\binom{2j+1}{j}})$ in $\{j\}$-CSP$(\bK_{2j+1})$, Adversary may be
seen to take on the role of denying $\bK_{\binom{2j+1}{j}} \models \Psi$ while
Prover is asserting that it is true. Thus, Adversary may always be assumed to
play variables $v^1,\ldots,v^j$ such that $|\{v^1,\ldots,v^j\}|=j$, because
otherwise he is simply making the job of Prover easier (by the properties of the
gadget $\mathcal{G}_j$). The behaviour of existential quantification in the
simulation is easy to see, but we will consider more carefully the behaviour of
universal quantification. The additional
$v^{1,1},\ldots,v^{1,j+1},\ldots,v^{j,1},\ldots,v^{j,j+1}$ force that every
possible subset $\{a_1,\ldots,a_j\} \subset \{1,\ldots,2j+1\}$ can be forced by
Adversary on $v^1,\ldots,v^j$. Indeed, Adversary may force any single element on
$v^i$ by avoiding it in $v^{i,1},\ldots,v^{i,j+1}$.

($\Rightarrow$) Assume $\bK_{\binom{2j+1}{j}} \models \Psi$. However Prover plays
the variables in $\Psi'$ corresponding to universal variables of $\Psi$, she
will be able to find a witness set $\pi^{-1}(a)$ for the variables in $\Psi'$
corresponding to an existential variable $x$ in $\Psi$, precisely because that
existential variable has some witness $a \in K_{\binom{2j+1}{j}}$. 

($\Leftarrow$) Assume $\bK_{2j+1} \models \Psi'$. No matter how Prover plays to
win $\mathscr{G}(\Psi',\bK_{2j+1})$, she will have possible witnesses sets
$\{a_1,\ldots,a_j\}$ for variables $\{v^1,\ldots,v^j\}$ in $\Psi'$ corresponding
to an existential variable $v$ of $\Psi$, for all sets $\{b_1,\ldots,b_j\}
\subset \{1,\ldots,2j+1\}$ corresponding to universal variables
$\{u_1,\ldots,v_j\}$ of $\Psi$ (because of the behaviour of the universal
variable simulation). Thus the existential witness $\pi(\{a_1,\ldots,a_j\})$ may
be used in $\Psi$ for $v$, and the result follows. 
\end{proofof}\vspace{-1ex}

\begin{proofof}{the claims (1a-c) from Proposition \ref{prop:cycles-i}}

For (1a), let $y$ be a predecessor of $x$. Then for the value $i$ chosen by 
Adversary for $y$, Prover must offer a set of at least three vertices of
$\bC_n$ that are adjacent to $i$ in $\bC_n$. Since there are only two such vertices,
Adversary can always choose for $x$ a vertex non-adjacent to $i$ at which
point Prover loses.

For (1b), let $y$ be the first vertex of the connected component of ${\cal
D}_\psi$ that contains $x$. Assume $y\neq x$ and consider the path ${\cal P}$ between
$x$ and $y$ in ${\cal D}_\psi$.  Without loss of generality, assume that the
value $i$ chosen by Adversary for $y$ is even. Note that, because $n$ is
even, if the length of ${\cal P}$ is also even, then Adversary must choose an even
value for $x$, while if the length is odd, she must choose an odd value
(otherwise Prover loses).  However, as $j>n/2$, the set provided by 
Prover for $x$ contains both an even and an odd number.  Thus Adversary is
allowed to choose for $x$ the wrong parity and Prover loses.

For (1c), suppose that $y$ and $z$ with $y\prec z$ are predecessors of $x$ where
$z$ is quantified by $\exists^{\geq j'}$ for some $j'\geq 2$.  If $i$ is the
value chosen by Adversary for $y$, then Prover must offer for $z$ a set
of $j'\geq 2$ values which hence must contain at least one value different from
$i$.  Adversary then chooses this value $i'$ after which Prover must
offer for $x$ two distinct vertices $i'',i'''$ of $\bC_n$ adjacent to both $i$ and
$i'$. But then $i,i',i'',i''$ yield a 4-cycle in $\bC_n$, impossible if $n\neq 4$.
\end{proofof}

\begin{proofof}{the case of odd $n$ in Proposition \ref{prop:cycles-iii}}
We argue that $\bC_n\models\Psi$ if and only if $\bC_n\models\Psi'$.  To do
this, it suffices to show that $\Psi'$ correctly simulates the universal
quantifiers of $\Psi$.  Namely, it suffices to prove that $\bC_n\models
Q_x\pi_x$, and for each $\ell\in\{0\ldots n-1\}$, Adversary has a strategy on
$Q_x\pi_x$ that evaluates $x$ to $\ell$.

For the first part, we provide a strategy for Prover.  We treat $x$ as
$x^{n+1}_1$.  For $x^1_1$, Prover offers any set.  For $x^k_1$ where $k\geq
2$, let $i$ be the value chosen by Adversary for $x^{k-1}_1$. By Lemma
\ref{lem:additive}, we observe that there are exactly $j$ vertices in $\bC_n$
having a walk to $i$ of length $j-1$. Prover offers this set for $x^k_1$.
This allows her to choose values for $x^{k-1}_2\ldots x^{k-1}_{j-1}$ as the path
$x^{k-1}_1,\ldots x^{k-1}_{j-1},x^k_1$ encodes precisely the fact that there
exists a walk of length $j-1$ between the values chosen for $x^{k-1}_1$ and
$x^k_1$. Thus $\bC_n\models Q_x\pi_x$.

For the second part, consider any $\ell\in\{0\ldots n-1\}$.  We explain a
strategy for Adversary that allows him to choose $\ell$ for $x$.  First, 
Adversary chooses any value for $x^1_1$.  Let $i_0$ be this value, and by the
second part of Lemma \ref{lem:additive}, choose a sequence of $n$ numbers
$i_1,i_2,\ldots,i_n$ either all from $\{-1,+1\}$ if $j$ is odd, or all from
$\{-2,0,+2\}$ if $j$ is even, such that $i_0+i_1+i_2+\ldots+i_n=\ell$. After
this, consider inductively $k\geq 2$ and let $i$ be the value chosen by 
Adversary for $x^{k-1}_1$.  By Lemma \ref{lem:additive}, there are exactly $j$
possible values that Prover can offer if she does not want to lose. Thus
Prover is forced to offer all these values. In particular, if $j$ is even,
this set contains values $i+1$ and $i-1$ (mod $n$) while if $j$ is odd, the set
contains values $i+2$, $i$, and $i-2$ (mod $n$). Thus Adversary is allowed
to choose the value $i+i_{k-1}$ (mod $n$) for $x^k_1$.  This shows that 
Adversary is allowed to choose the value $i_0+i_1+\ldots+i_n=\ell$ for
$x^{n+1}_1=x$.

Thus, this proves that $\Psi'$ correctly simulates the universal quantifiers of
$\Psi$, and consequently $\bC_n\models\Psi$ if and only if $\bC_n\models\Psi'$.
For odd $n$, this completes the proof of the claim that $\{1,j\}$-CSP($\bC_n$)
is Pspace-hard.
\end{proofof}

\begin{proofof}{the case of even $n$ in Proposition \ref{prop:cycles-iii}}

We analyse possible assignments to the vertices $v_0,\ldots,v_{n-1}$.  For
clarity, we define $\alpha_k=kj-r-k-1$ and note that $n/2+1=\alpha_k$ for
$k=\big\lceil\frac{n+4}{2(j-1)}\big\rceil$.  By the symmetry of $\bC_n$, we
assume that Adversary chooses for $w_0$ the value $n-r-1$.  The next quantified
vertex is $v_{j-r-2}=v_{\alpha_1}$ in distance $j-1$ from $w_0$.  Thus, by Lemma
\ref{lem:additive}, there are exactly $j$ values that Prover can and must offer.
Among them, we find $j-r-2=\alpha_1$.  Similarly, for $2\leq k\leq
\big\lceil\frac{n+4}{2(j-1)}\big\rceil$, the vertex $v_{\alpha_{k-1}}$ is in
distance $j-1$ from $v_{\alpha_k}$, and hence, Prover is forced to offer a set
of $j$ values only depending on the value chosen for $v_{\alpha_{k-1}}$.  In
particular, if $\alpha_{k-1}$ was chosen for $v_{\alpha_{k-1}}$, then Adversary
can choose $\alpha_k$ for $v_{\alpha_k}$.  This argument also shows that if
Prover acts as we describe, then in every possible case she can complete the
homomorphism for the path $w_0,w_1,\ldots,w_r,v_0\ldots,v_{n/2+1}$.  Further,
she also has a way of assigning the values to $v_{n/2+2},\ldots,v_{n-1}$. This
can be seen as follows.  First, note that the distance betwen $v_{n/2+1}$ and
$v_0$ is $n/2-1$.  Thus, if $n/2$ is odd, then the values assigned to $v_0$ and
$v_{n/2+1}$ have the same parity because $n/2+1$ is even and we observe that
between any two vertices of the same parity in $\bC_n$ there exists a walk of
length $n/2-1$. Similarly, if $n/2$ is even, the values chosen for
$v_0,v_{n/2+1}$ have different parity and between any two vertices of $\bC_n$ of
different parity there is a walk of length $n/2-1$.

This concludes the argument for the vertices $v_0,\ldots,v_{n-1}$.  It implies
two possible types of outcomes:  either the values chosen for $v_0, \ldots,
v_{n-1}$ are all distinct, or  not.  To obtain the former case, for each $1\leq
k \leq \big\lceil \frac{n+4}{2(j-1)} \big\rceil$, Adversary chooses
$\alpha_k$ for $v_{\alpha_k}$. This forces assigning $i$ to $v_i$ for all
$i\in\{0, \ldots, n/2+1\}$ and thus consequently also for all the other $v_i$'s.
We shall assume this situation first. For all other (degenerate) cases we use a
different argument explained later.

Thus assuming that $v_0,\ldots,v_{n-1}$ get assigned values $0,\ldots,n-1$ in
that order and the values for the original variables of $\Psi$ are chosen, we
argue that Prover can finish the homomorphism if and only if the assignment
to the variables of $\Psi$ is a proper colouring for ${\cal D}_\psi$.  This
follows exactly as in \cite{FederHellHuang99}. Namely, in every gadget, each copy
of $\bC_n$ is forced to copy the assignment from the adjacent copy of $\bC_n$,
shifted by $+1$ or by $-1$ (mod $n$).  In particular, if $i$ is the value
assigned to~$y$, the vertex $z$ opposite $y$ in the last copy of $\bC_n$ is
necessarily assigned value $n/2+i$ (mod $n$).  This implies that the value
assigned to $x$ is different from $i$ as the path from $z$ to $x$ is too short
(of length less than $n/2$).  On the other hand, this path is long enough so
that any value of the same parity as $i$ but different from $i$ can be chosen
for $x$ such that the homomorphism can be completed.  This precisely simulates
the edge predicate of $\Psi$.  Finally, we observe that Prover can choose
whether consecutive copies of $\bC_n$ are shifted by $+1$ or $-1$ and there are
exactly $3n/2$ copies of $\bC_n$. Thus, by Lemma \ref{lem:additive}, every possible
odd number from $\{0,\ldots,n-1\}$ can be chosen for $y$ by a particular series
of shifts. It follows that $\{1,3,\ldots,n-1\}$ is precisely the set colours we use to
simulate QCSP($\bK_{n/2}$).
\medskip

Now, we discuss the degenerate cases.  Namely, we show that, regardless of the
assignment to $v_0,\ldots,v_{n-1}$, for each copy of the gadget (in Figure
\ref{fig:even-case}) there is a way to complete the homomorphism (by assigning
the values to the white vertices) in such a way that if $\ell$ is the value
assigned to~$y$, then the vertex opposite $y$ in the last copy of $C_n$ is
assigned value $\ell+n/2$ (mod $n$).  As this is exactly what happens in the
non-degenerate case, the rest will follow. Note that consideration of these degenerate cases is the reason we use a chain of $3n/2$ copies of $\bC_n$ in the gadget of Figure~\ref{fig:even-case}, instead of the $n/2$ used in the like gadget in \cite{FederHell98}.

Since we assume that the vertices $v_0,\ldots,v_{n-1}$ are assigned a proper
subset of $\{0,\ldots,n-1\}$, it can be seen that they are assigned a circularly
consecutive subset of these numbers, and this subset is of size at most $n/2+1$
as otherwise the assignment cannot be a homomorphism. (Recall that in the proof
we argue that we can assume that the assignment to $v_0,\ldots,v_{n-1}$ is a
homomorphism).

For simplicity, let $\lambda$ denote the assignment constructed so far, i.e., a
mapping from the assigned vertices to their assigned values. We explain how to
complete this assignment for the gadget so that it becomes a homomorphism to
$C_n$.

The gadget contains $3n/2$ copies of $C_n$. We consider them from the right to
left, namely $\{v_0,\ldots,v_{n-1}\}$ is the 1st copy, and the $3n/2$-th copy is
the one containing $y$.  With this in mind, we denote by
$v_0^i,\ldots,v_{n-1}^i$ the respective copies of $v_0,\ldots,v_{n-1}$ in the
$i$-th copy of $C_n$.  In particular, $y$ is the vertex $v^{3n/2}_{n/2}$.

We describe the assignment to the copies of $C_n$ in three phases.  In the first
phase, we assign values to the first $n/2$ copies.  Consider $1\leq i<n/2$, and
assume that the vertices $v_0^i,\ldots,v_{n-1}^i$ are assigned values between
$a$ and $b$ (inclusive) in the clock-wise order.  Then the assignment to the
$(i+1)$-st copy of $C_n$ is as follows.  For $k\in\{0,\ldots,n-1\}$, if
$\lambda(v_k^i)=a$, then we set $\lambda(v_k^{i+1})=a+1~({\rm mod}~n)$,
otherwise we set $\lambda(v_k^{i+1})=\lambda(v_k^i)-1~({\rm mod}~n)$.  It is
easy to verify that this constitutes a homomorphism to $C_n$. It follows that
$\big|\lambda\big(\{v^{n/2}_0,\ldots,v^{n/2}_{n-1}\}\big)\big|=2$.

Next, we explain the assignment to the second $n/2$ copies of $C_n$.  Let $\ell$
be the value assigned to $y$.  We choose the values for the second $n/2$ copies
in such a way that consecutive copies of $C_n$ are just shifted by $+1$ or $-1$. We can
choose an appropriate sequence of $+1,-1$ shifts so that the value assigned to
$v^n_{n/2}$ is exactly $\ell+n/2~({\rm mod}~n)$. (The argument about the parity
of these values is the same as in the non-degenerate case.) We further
conclude that $\big|\lambda\big(\{v^{n}_0,\ldots,v^{n}_{n-1}\}\big)\big|=2$.

The assignment to the final $n/2$ copies is as follows.  For $n\leq i<3n/2$,
again assume that the vertices $v_0^i,\ldots,v_{n-1}^i$ are assigned values
between $a$ and $b$ in the clockwise order.  Then for
$k\in\{0,\ldots,n-1\}\setminus\{n/2\}$, if $\lambda(v_k^i)=a+1~({\rm mod~}n)$
and $\lambda\big(v_{(k-1)~{\rm mod}~n}^i\big)=\lambda\big(v_{(k+1)~{\rm
mod}~n}^i\big)=a$, then we set $\lambda(v_k^{i+1})=a$, and otherwise we set
$\lambda(v_k^{i+1})=\lambda(v_k^i)+1$.  Again, we conclude that this constitutes
a homomorphism, and it follows that
$\big|\lambda\big(\{v^{3n/2}_0,\ldots,v^{3n/2}_{n-1}\}\big)\big|=n/2+1$.  In
particular,  we observe that $\lambda\big(y=v^{3n/2}_{n/2}\big)=\ell$ and
$\lambda\big(v^{3n/2}_0\big)=\ell+n/2~({\rm mod}~n)$.

That concludes the argument.
\end{proofof}\vspace{-1ex}

\begin{proofof}{Proposition \ref{prop:c4*}}
We will reduce from the problem QCSP$(\mathcal{K}_4)$ (known to be
Pspace-complete from, e.g., \cite{BBCJK}). We will borrow heavily from the
reduction of CSP$(\mathcal{K}_4)$ to Ret$(\bC^*_4)$ in \cite{FederHell98}. The
reduction has a very similar flavour to that used in Case $(ii)$ of
Theorem~\ref{thm:cycles}, but borrows from \cite{FederHell98} instead of
\cite{FederHellHuang99}.

For an input $\Psi:=Q_1 x_1 Q_2 x_2 \ldots Q_m x_m \ \psi(x_1,x_2,\ldots,x_m)$
for QCSP$(\bK_4)$ we build an input $\Psi'$ for $\{1,2,3,4\}$-CSP$(\bC^*_4)$ as
follows. We begin by considering the graph $\mathcal{D}_\psi$, from which we
first build a graph $\mathcal{G}':=\mathcal{D}_\psi \uplus \mathcal{C}^*_4$. Now
we build $\mathcal{G}''$ from $\mathcal{G}'$ by replacing every edge $(x,y) \in
\mathcal{D}_\psi$ with the following gadget (which connects also to the fixed
copy of $\mathcal{C}^*_4$ in $\mathcal{G}'$ -- induced by $\{z_1,\ldots,z_4\}$ -- as
drawn in the picture).

\begin{figure}[h!]
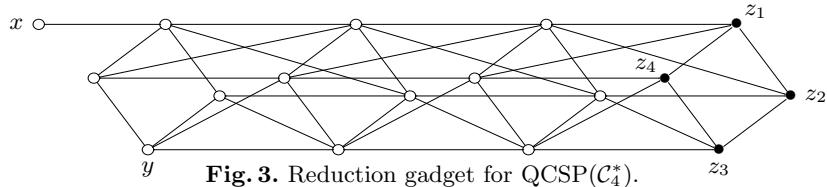

\vspace{-4ex}
\mbox{}\hfill$
\xy/r2pc/:
(9,0);(9,0);{\xypolygon4"X"{~>{-}~={82}[o]{\bullet}}};
(6,0);(6,0);{\xypolygon4"Y"{~>{-}~={82}[o][F]{\phantom{s}}}};
(3,0);(3,0);{\xypolygon4"W"{~>{-}~={82}[o][F]{\phantom{s}}}};
(0,0);(0,0);{\xypolygon4"Z"{~>{-}~={82}[o][F]{\phantom{s}}}};
{\ar@{-} "X1";"Y1"};
{\ar@{-} "X1";"Y2"};
{\ar@{-} "X2";"Y2"};
{\ar@{-} "X2";"Y3"};
{\ar@{-} "X3";"Y3"};
{\ar@{-} "X3";"Y4"};
{\ar@{-} "X4";"Y4"};
{\ar@{-} "X4";"Y1"};
{\ar@{-} "Y1";"W1"};
{\ar@{-} "Y1";"W2"};
{\ar@{-} "Y2";"W2"};
{\ar@{-} "Y2";"W3"};
{\ar@{-} "Y3";"W3"};
{\ar@{-} "Y3";"W4"};
{\ar@{-} "Y4";"W4"};
{\ar@{-} "Y4";"W1"};
{\ar@{-} "W1";"Z1"};
{\ar@{-} "W1";"Z2"};
{\ar@{-} "W2";"Z2"};
{\ar@{-} "W2";"Z3"};
{\ar@{-} "W3";"Z3"};
{\ar@{-} "W3";"Z4"};
{\ar@{-} "W4";"Z4"};
{\ar@{-} "W4";"Z1"};
"Z1"+(-2,0)*[o][F]{\phantom{s}}="a1";
{\ar@{-} "a1";"Z1"};
"a1"+(-0.35,0)*{x};
"Z3"+(0,-0.3)*{y};
"X1"+(0.3,0.2)*{{z_1}};
"X4"+(0.4,0)*{{z_2}};
"X3"+(0,-0.3)*{{z_3}};
"X2"+(-0.3,0.2)*{{z_4}};
\endxy$\hfill\mbox{}
\vspace{-4ex}
\caption{Reduction gadget for QCSP$(\mathcal{C}^*_4)$.\label{fig:c4star}}
\vspace{-2ex}
\end{figure}

\noindent $\phi_{\mathcal{G}''}$ will form the quantifier-free part of $\Psi'$;
we now explain the structure of the quantifiers. Let $z_1,\ldots,z_4$ correspond
in $\phi_{\mathcal{G}''}$ to the fixed copy of $\mathcal{C}^*_4$ in
$\mathcal{G}''$. $\Psi'$ begins $\exists z_1 \exists^{\geq 2} z_2 \exists^{\geq
3} z_3 \exists^{\geq 2} z_4$ ($z_1$ could equally be quantified with
$\exists^{\geq j}$, $j>1$). Now we continue in the quantifier order of $\Psi$.
When we meet an $\exists$ quantifier, we quantify with $\exists$ the
corresponding vertex in $\phi_{\mathcal{G}''}$. When we meet a $\forall$
quantifier, we quantify with $\forall=\exists^{\geq 4}$ the corresponding vertex
in $\phi_{\mathcal{G}''}$. Finally, we quantify with $\exists$ all remaining
variables, corresponding to vertices we added in gadgets in $\mathcal{G}''$. We
claim that $\bK_3 \models \Psi$ iff $\bC^*_4 \models \Psi'$. The proof of this
proceeds as with Theorem~\ref{thm:cycles} (though there are several more
degenerate cases to consider). 
\end{proofof}

\begin{proofof}{Corollary \ref{cor:c4*}}
\,
We give a reduction from $\{1,2,3,4\}$-CSP$(\bC^*_4)$ to QCSP$(\bC^*_4)$, using
the following shorthands ($x',x''$ must appear nowhere else in $\phi$, which may
contain other free variables).
\[
\begin{array}{l}
\exists^{\geq 1} x \ \phi(x) := \exists x \ \phi(x) \\
\exists^{\geq 2} x \ \phi(x) := \forall x' \exists x \ E(x',x) \wedge \phi(x) \\ 
\exists^{\geq 3} x \ \phi(x) := \forall x'' \forall x' \exists x \ E(x'',x) \wedge E(x',x) \wedge \phi(x) \\ 
\exists^{\geq 4} x \ \phi(x) := \forall x \ \phi(x) \\ 
\end{array}
\]
On $\mathcal{C}^*_4$, it is easy to verify that, for each $i \in [4]$,
$\exists^i x \ \phi(x)$ holds iff there exist at least $i$ elements $x$
satisfying $\phi$. The result follows easily (note that each use of shorthand
substitution involves new variables corresponding to $x$ and $x'$ above).

That concludes the proof.
\end{proofof}

\subsection{Example for Corollary \ref{cor:twins}}

As an example of why the construction of Corollary~\ref{cor:twins} was
necessary, consider the graph $\bH^{\mathrm{hairy}}_6$, depicted in
Figure~\ref{fig:hairy} on the right,  with $18$ vertices made from $\bC_6$ with
twin paths of length one added to each vertex. The first part of $\Psi'$ takes
the form depicted in Figure \ref{fig:hairy} on the left.

\begin{figure}[h!]
\vspace{-6ex}
\mbox{}\hfill
$\xy/r1.7pc/:
(0,0)*[o][F]{\phantom{s}}="v1";
(2,1)*[o][F]{\phantom{s}}="x1";
(2,0)*[o][F]{\phantom{s}}="v2";
(4,1)*[o][F]{\phantom{s}}="x1'";
(4,0)*[o][F]{\phantom{s}}="v3";
(4,-1)*[o][F]{\phantom{s}}="x2";
(6,1)*[o][F]{\phantom{s}}="x3";
(6,0)*[o][F]{\phantom{s}}="v4";
(6,-1)*[o][F]{\phantom{s}}="x2'";
(8,1)*[o][F]{\phantom{s}}="x3'";
(8,0)*[o][F]{\phantom{s}}="v5";
(10,0)*[o][F]{\phantom{s}}="v6";
{\ar@{-} "v1";"x1"};
{\ar@{-} "v1";"v2"};
{\ar@{-} "x1";"x1'"};
{\ar@{-} "v2";"v3"};
{\ar@{-} "v2";"x2"};
{\ar@{-} "x1'";"v4"};
{\ar@{-} "v3";"v4"};
{\ar@{-} "v3";"x3"};
{\ar@{-} "x3";"x3'"};
{\ar@{-} "x2";"x2'"};
{\ar@{-} "x2'";"v5"};
{\ar@{-} "v4";"v5"};
{\ar@{-} "x3'";"v6"};
{\ar@{-} "v5";"v6"};
"v1"+(0,-0.3)*{v_1};
"x1"+(0,0.3)*{x_1};
"x1'"+(0,0.35)*{x_1'};
"v2"+(0,-0.3)*{v_2};
"x2"+(0,-0.3)*{x_2};
"x2'"+(0,-0.3)*{x_2'};
"v3"+(0,-0.3)*{v_3};
"v4"+(0,-0.3)*{v_4};
"x3"+(0,0.3)*{x_3};
"x3'"+(0,0.35)*{x_3'};
"v5"+(0,-0.3)*{v_5};
"v6"+(0,-0.3)*{v_6};
\endxy
$
\qquad
$\xy/r1.7pc/:
(0,0);(0,0);{\xypolygon6"X"{~>{-}~={0}[o][F]{\phantom{s}}}};
(0,0);(0,0);{\xypolygon12"Y"{~:{(2,0):}~>{}~={15}[o][F]{\phantom{s}}}};
{\ar@{-} "X1";"Y1"};
{\ar@{-} "X1";"Y12"};
{\ar@{-} "X2";"Y3"};
{\ar@{-} "X2";"Y2"};
{\ar@{-} "X3";"Y5"};
{\ar@{-} "X3";"Y4"};
{\ar@{-} "X4";"Y7"};
{\ar@{-} "X4";"Y6"};
{\ar@{-} "X5";"Y9"};
{\ar@{-} "X5";"Y8"};
{\ar@{-} "X6";"Y11"};
{\ar@{-} "X6";"Y10"};
\endxy$
\hfill\mbox{}
\vspace{-2ex}
\caption{Finding $\bC_6$ in $\bH^{\mathrm{hairy}}_6$.\label{fig:hairy}}
\vspace{-4ex}
\end{figure}

\noindent In this case, each of $v_1,v_2,v_3,v_4,x'_1,x_1$ or
$v_2,v_3,v_4,v_5,x'_2,x_2$ must be evaluated to $\bC_6$
($v_3,v_4,v_5,v_6,x'_3,x_3$ need not). For an example where
$v_1,v_2,v_3,v_4,x'_1,x_1$ may fail to be evaluated to a $\bC_6$, but a later
cycle does not, consider two disjoint copies of $\bC_6$ joined by a path of
length two.

\subsection{Tightness of Proposition~\ref{prop:2ton-3}}
\vspace{-1ex}

\begin{proposition}
\label{prop:13P5}
$\{1,3\}$-CSP$(\bP_{5})$ is in $\mathrm{L}$.
\end{proposition}\vspace{-1ex}

\begin{proof}
We will consider an input $\Psi$ to $\{1,3\}$-CSP$(\bP_{5})$ of the form $Q_1 x_1 Q_2 x_2 \ldots$ $Q_m x_m \ \psi(x_1,x_2,\ldots,x_m)$. An instance of an $\exists^{\geq 3}$ variable is called \emph{trivial} if it has neither a path to another (distinct) $\exists^{\geq 3}$ variable, nor a path to an $\exists$ variable that precedes it in the natural order on $\mathcal{D}_\psi$. A minimum requirement on a yes-instance $\Psi$ is that $\mathcal{D}_\psi$ be bipartite, which can be determined in L; henceforth we assume this. Note that in the game $\mathscr{G}(\Psi,\bP_{5})$, for $\exists$ variables, Prover advances a singleton set (and Adversary must choose that vertex); \mbox{i.e.} effectively only Prover plays.

In this case one partition is a vertex larger than the other, and this partition contains the vertices at both ends. Any non-trivial instance of an $\exists^{\geq 3}$ variable must be evaluable on all vertices of the larger partition. If in $\mathcal{D}_\psi$ there are two $\exists^{\geq 3}$ vertices with a path of odd length or length $< 4$ between them, then $\Psi$ must be a no-instance (recall these variables must be evaluable on all vertices of the larger partition, at extreme ends these are at distance $4$). Further, if in $\mathcal{D}_\psi$ there is an $\exists$ vertex preceding in the order an $\exists^{\geq 3}$ vertex, and these are joined by a path of length $< 2$, then $\Psi$ must be a no-instance. Otherwise, we claim $\Psi$ is a yes-instance, and Prover may witness this by using the following \emph{centre-finding} strategy. Let $3$ be the centre vertex of $\bP_{5}$ and let $2$ be one of its neighbours. Essentially, Prover is always trying to move towards $3$ and $2$ on the $\exists$ variables (recall she always suggests the three variables in the larger bipartition as the witness set for an $\exists^{\geq 3}$ variable). When given an $\exists$ variable, corresponding to a vertex $x$ in $\mathcal{D}_\psi$, to evaluate, Prover must look at all vertices at distance $\leq 2$ from $x$ in $\mathcal{D}_\psi$. If some of these are already evaluated, then Prover looks at the closest she can get to $3$ and $2$ with $x$ (bearing parity in mind). This will result in $x$ being played on one of the vertices $2,3,4$; and on $4$ only if $x$ is adjacent to a vertex already played on $5$. Otherwise, if none of these are already played, then Prover looks to see if there is a path in $\mathcal{D}_\psi$ to either an $\exists^{\geq 3}$ vertex, as yet unplayed, or any other vertex already played. If there is such a path to an already played vertex, then she plays on $3$ or $2$ according to the parity of the played vertex and length of the path. If there is such a path to an unplayed $\exists^{\geq 3}$ vertex, and it is of odd length, Prover plays $x$ on $2$. In all other cases, Prover plays $x$ on $3$. As this always provides a vertex on which Prover may play, it is seen to be a winning strategy for her.
\end{proof}
\vspace{-2ex}

\begin{proposition}
Let $\bH$ be bipartite. Then $\{1,j\}$-CSP$(\bH)$ is in $\mathrm{L}$ if $j \in \{1,|H|-2,|H|-1,|H|\}$.
\end{proposition}\vspace{-1ex}

\begin{proof}
When $j:=|H|$ we have QCSP$(\bH)$ and we refer for the result to \cite{CiE2006}. For $j:=|H|-1$, we argue as in the proof of Proposition~\ref{prop:bipartition} unless $\bH$ is the complete bipartite (star) $\bK_{1,l}$ (for some $l$), in which case we appeal to Proposition~\ref{prop:complete-bipartite}. The case $j:=|H|-2$ is not much more complicated. If we do not fall as in the proof of Proposition~\ref{prop:bipartition} or under Proposition~\ref{prop:complete-bipartite}, then we are equivalent to $\{1,3\}$-CSP$(\bP_5)$ and the result follows from Proposition~\ref{prop:13P5}.
\end{proof}

\end{document}